\documentclass[12pt,a4paper]{article}

\usepackage{amssymb}
\usepackage{amsmath, amsthm}
\usepackage{arydshln}
\usepackage{epsfig}
\usepackage{setspace}
\usepackage{comment}
\usepackage{geometry}
\usepackage{hyperref}
\usepackage{mathtools}
\usepackage{mleftright}\mleftright

\def\RR{\mathbb{R}}

\def\ZZ{\mathbb{Z}}
\def\QQ{\mathbb{Q}}

\def\KK{\mathbb{K}}
\def\11{\mathbf{1}}

\numberwithin{equation}{section}

\newcommand{\supp}{\mathop{\rm supp} }

\newcommand{\proj}{\mathop{\rm proj}}

\newcommand{\image}{{\rm Im}}

\newtheorem{Thm}{Theorem}[section]
\newtheorem{Prop}[Thm]{Proposition}
\newtheorem{Lem}[Thm]{Lemma}

\theoremstyle{definition}

\title{Horospherically convex optimization for  
fractional subspace packing and its applications}
\author{Hiroshi Hirai}

\begin{document}

\maketitle
\begin{abstract}
In this paper, we address a semi-infinite LP relaxation of the vector-subspace packing problem.
This is a higher-dimensional generalization of the fractional linear matroid parity problem 
and is closely related to Brascamp-Lieb polytopes.
We show that the dual of this LP can be formulated as ``linear programming on a Euclidean building," namely, the problem of minimizing a Busemann function over an intersection of horoballs.
This provides a natural example of horospherically convex optimization, recently introduced by Goodwin et al. (2026) and Criscitiello and Kim (2025).
By applying the incremental Busemann subgradient method, 
we obtain an additive FPTAS for the problem.
As applications, we obtain a new and simpler polynomial-time algorithm
for fractional linear matroid parity, and new algorithms
for the membership problem of Brascamp-Lieb polytopes.
\end{abstract}

\noindent
MSC-classifications: 90C34, 68W40

\section{Introduction}
A vector-space generalization of the well-known {\em set packing problem} 
and its dual, the {\em hitting set problem} (or {\em set cover problem}) can be formulated as follows: 
Let $\KK$ be a field. We are given a collection of vector subspaces $S_1,S_2,\ldots,S_m$ in $\KK^n$. 
\begin{description}
    \item[Subspace packing problem:] 
    Find a subset $I \subseteq \{1,2,\ldots,m\}$ of indices with the maximum cardinality
    such that the vector-space sum $\sum_{k \in I}S_k$ is equal to the direct sum $\bigoplus_{k \in I}S_k$.
    \item[Hitting subspace problem:]
     Find a vector subspace $X$ in $\KK^n$ with the minimum dimension
     such that it
     intersects all $S_k$ nontrivially $(S_k \cap X \neq \{{\bf 0}\})$.
\end{description}
When each $S_k$ is a coordinate subspace, 
they are the set packing and hitting set problems, respectively.
As in this case,  
the two problems are in weak duality $\max_I |I| \leq \min_X \dim X$.
When  $\KK^n = U \oplus V$ and $S_k$ is the 2-dimensional subspace spanned by $u_k \in U, v_k \in V$ for $k=1,2,\ldots,m$, the subspace packing problem is the linear matroid intersection problem, 
and the strong duality is implied by the matroid intersection theorem.
For general $2$-dimensional subspaces $S_k$, though the strong duality does not hold,
the problem is known as the linear matroid parity problem and a famous result by Lov\'asz~\cite{Lovasz80} says that it is polynomial-time solvable.
When each $S_k$ is $d$-dimensional, 
the problem is viewed as the matroid $d$-parity problem~\cite{LeeSviridenkoVondrak2013} on a linearly represented matroid.
For the hitting subspace problem,
we could not find any work from an optimization point of view.
This may be because it appears to be an intractably difficult infinite discrete optimization problem.
We found one paper~\cite{Gadzinski} which studies this problem 
from different perspectives of low-rank matrix completion; 
the name is taken from that paper.

The focus of the present paper is their fractional relaxations.
If $\sum_{k \in I}S_k = \bigoplus_{k \in I}S_k$, then 
for every vector subspace $X$ it holds $\sum_{k \in I}S_k \cap X = \bigoplus_{k \in I}S_k \cap X$, 
and hence $\sum_{k \in I} \dim S_k \cap X = \dim \sum_{k \in I}S_k \cap X \leq \dim X$.
This leads to the following semi-infinite LP relaxation of the subspace packing problem:
\begin{description}
    \item[Fractional subspace packing problem:]
    \begin{eqnarray}
\mbox{Max.} && \sum_{k=1}^m x_k \nonumber \\
\mbox{s.t.} && \sum_{k=1}^m x_k \dim S_k \cap X \leq \dim X \quad ( X \in {\cal S}(\KK^n)), \nonumber\\
&& x_k \geq 0 \quad (k=1,2,\ldots,m), \label{eqn:fractional_subspace_packing}
\end{eqnarray}
\end{description}
where ${\cal S}(\KK^n)$ denotes the set of all vector subspaces of $\KK^n$.
This LP was introduced by Vande Vate~\cite{VandeVate92} 
as an LP relaxation of the matroid parity problem, 
where a feasible solution is a {\em fractional matroid matching}.
It was also considered by Lee, Sviridenko, and Vondr\'ak~\cite{LeeSviridenkoVondrak2013} for matroid $d$-parity.
Notice that the matroid for our setting 
is the (infinite) matroid of all vectors of $\KK^n$. 
As we will explain later, 
LP (\ref{eqn:fractional_subspace_packing}) has renewed interest due to 
its connections to the {\em noncommutative rank (nc-rank)}~\cite{IQS15a} and the {\em Brascamp-Lieb inequality}~\cite{BrascampLieb,Lieb}, 
which are the main motivation for this work.

On the other hand, from $X \cap S_k \neq \{{\bf 0}\} \Leftrightarrow \dim X \cap S_k \geq 1$, 
an LP relaxation of the hitting subspace problem is naturally obtained as:
\begin{description}
    \item[Fractional hitting subspace problem:]
    \begin{eqnarray}
\mbox{Min.} && \sum_{X \in {\cal S}(\KK^n)} \lambda(X) \dim X \nonumber\\
\mbox{s.t.} && \sum_{X \in {\cal S}(\KK^n)} \lambda(X) \dim S_k \cap X \geq 1 \quad (k \in [m]), \nonumber\\
&& \lambda: {\cal S}(\KK^n) \to \RR_{\geq0},\ |\supp \lambda| < \infty,
\label{eqn:fractional_hitting_subpace}
\end{eqnarray} 
\end{description}
where  $\supp \lambda$ denotes the nonzero support $\{X \in {\cal S}(\KK^n)\mid \lambda(X) \neq 0\}$.
We call a feasible solution a {\em fractional hitting subspace}.
Observe that (\ref{eqn:fractional_hitting_subpace}) is LP-dual to (\ref{eqn:fractional_subspace_packing}). 
Furthermore, the strong duality holds, 
as they are implicitly finite-dimensional LPs. Indeed, the possible coefficients of inequalities of these two LPs are finite:
\begin{equation}\label{eqn:coefficient_vector}
(\dim X, \dim S_1 \cap X, \ldots, \dim S_m \cap X) \in \{0,1,2\ldots,n\}^{m+1}.
\end{equation}
The optimal value (the {\em fractional packing/hitting number}) is denoted by $\tau^*$.
Although they reduce to finite-dimensional LPs, 
executing the reduction is extremely difficult; determining the possible coefficients 
requires to solve algebraic equations exponentially many times. 
We know nothing about 
the computational complexity of these LPs, 
except for the $2$-dimensional case, namely, the fractional linear matroid parity problem.
For this case, Chang, Llewellyn, and Vande Vate~\cite{CLV01a,CLV01b} gave a polynomial-time algorithm ({\em CLV-algorithm}), 
and Oki and Soma~\cite{OkiSoma_SICOMP} gave a faster randomized algorithm based on the connection to the nc-rank.

In this paper, we tackle the semi-infinite LP~(\ref{eqn:fractional_hitting_subpace}) for the general case of $S_k$, 
where computation on an abstract field $\KK$ is performed by an algebraic RAM. 
That is, we assume:
\begin{description} 
    \item[(A)] Every element in $\KK$ is kept in constant space, 
    and arithmetic operations $+,-,\times,/$ (and comparison) in $\KK$ are done in constant time.  
\end{description}

The main result of this paper 
is an additive FPTAS for the fractional hitting subspace problem (\ref{eqn:fractional_hitting_subpace}):
\begin{Thm}\label{thm:main}
For $\epsilon > 0$, we can compute in $O(n^5m^3 \epsilon^{-2})$ time a fractional hitting subspace $\lambda$ with
\begin{equation}
\sum_{X} \lambda(X)\dim X  \leq \tau^* + \epsilon.
\end{equation}
\end{Thm}
Since $\tau^* \geq 1$, this also yields a multiplicative FPTAS. 
Note that LP~(\ref{eqn:fractional_hitting_subpace}) is very fractional; 
it has no (half-)integrality property of extreme points.
To the best of our knowledge, there are no comparable results
on the computational complexity for such semi-infinite linear programming.

For proving this result, we utilize 
geodesically convex optimization 
on nonpositively curved spaces called {\em Hadamard spaces}. 
Deriving an explicit computational complexity
via optimization on such spaces was pioneered by the work of
Garg et al.\cite{GGOW15} who developed a polynomial-time algorithm for the nc-rank 
by {\em operator scaling}. 
It turned out that operator scaling is a convex optimization on the Hadamard manifold of positive definite matrices; see e.g., \cite{BFGOWW} for subsequent wonderful developments.
Another polynomial-time algorithm for nc-rank due to Hamada and Hirai~\cite{HamadaHirai} is based on convex optimization on a 
{\em Euclidean building}---a representative example of Hadamard spaces without manifold structure. 
Our approach to (\ref{eqn:fractional_hitting_subpace}) is fully inspired by this work.   
By the well-known uncrossing argument, 
the support of $\lambda$ in (\ref{eqn:fractional_hitting_subpace}) can be assumed to be a {\em flag}---a chain of vector subspaces.
Then, the (nonconvex) subset of feasible solutions with flag-support has the structure of a Euclidean building. 
Further, the objective function is geodesically convex on this space.
As done in~\cite{HamadaHirai}, 
the splitting proximal point algorithm by Ba\v{c}\'{a}k~\cite{Bacak} and Ohta-P\'alfia~\cite{OhtaPalfia} is applicable to obtain an approximate solution.  
However, due to the lack of the strong convexity of the objective, 
the desired complexity estimate cannot be deduced. 
Instead, we exploit another convexity property, 
called the {\em horospherical convexity} (or {\em Busemann subdifferentiability}).
This notion was introduced independently by Goodwin et al.~\cite{GoodwinLewisGenaroNicolae2026} and Criscitiello and Kim~\cite{CriscitielloKim2025} for 
overcoming the curvature obstruction in first-order methods in Hadamard spaces/manifolds; see also 
a preceding work~\cite{GoodwinLewisGenaroNicolae2026mor}.
Analogously to Euclidean convex functions, 
{\em horospherically convex functions} are defined as those supported everywhere 
by affine-function analogues in Hadamard spaces, called {\em Busemann functions}.
We show that (\ref{eqn:fractional_hitting_subpace}) is naturally formulated as ``{\em linear programming on  Euclidean building}"---the problem of 
minimizing a Busemann function over a convex set defined by Busemann functions (i.e., intersection of {\em horoballs}).
This is an ideal example of horospherically convex optimization 
for which the {\em incremental Busemann subgradient method}~\cite{GoodwinLewisGenaroNicolae2026} 
is applicable.
Then Theorem~\ref{thm:main} is obtained from the convergence result in \cite{GoodwinLewisGenaroNicolae2026}.
Since the derivation is surprisingly natural,
we believe that (\ref{eqn:fractional_hitting_subpace}) 
will serve as a standard example of horospherically convex optimization in the future.

The rest of this paper is organized as follows.
In Section~\ref{sec:preliminaries}, 
we introduce necessary terminologies, lemmas, and 
tools on Hadamard spaces, horospherically 
convex optimization, and Euclidean buildings.
In Section~\ref{sec:formulation&algorithm},  
we consider a slightly general version of (\ref{eqn:fractional_hitting_subpace}), 
formulate it as a horospherically convex optimization, 
and obtain a refined version of Theorem~\ref{thm:main}.
This also yields an FPT-algorithm computing $\tau$, where the dimension $n$ is the fixed parameter.
In Section~\ref{sec:application}, 
we present applications. The first application is 
a new deterministic polynomial-time algorithm for solving  
fractional linear matroid parity.
This is an immediate application of the incremental subgradient method and the half-integrality of the problem.
The resulting algorithm is slower than the CLV-algorithm~\cite{CLV01a,CLV01b} but considerably simpler.  
The second application concerns the computational complexity of {\em Brascamp-Lieb polytopes (BL-polytopes)}~\cite{BCCT08}, 
which is our motivating target.
The BL-polytope is the parameter space for which the BL-inequality is nontrivial.
Since the pseudo-polynomial time membership algorithm by Garg et al.~\cite{GGOW18},
the computational complexity of BL-polytopes has attracted interest in 
theoretical computer science.
As pointed out by Franks, Soma, and Goemans~\cite{FranksSomaGoemans2022}, 
the BL-polytope is the face of the feasible region of (\ref{eqn:fractional_subspace_packing})
determined by $\sum_{k=1}^m x_k \dim S_k \cap X = \dim X$ for $X = \KK^n (=\RR^n)$.
By using our results, we present new algorithms for 
approximate membership of BL-polytopes.
However, assumption (A) may not be appropriate for studying the computational complexity of BL-polytopes. 
Also our algorithms have no guarantee of polynomial bit-complexity when
dealing with rational inputs in the standard RAM model, unfortunately.
Nevertheless, they are conceptually different from existing operator-scaling-based algorithms~\cite{FranksSomaGoemans2022,GGOW18}, and are expected to stimulate future research.  
In the additional Section~\ref{sec:set_packing/cover}, we verify that our problems actually generalize fractional set packing and cover problems. 
Specifically, 
we show the not-quite-obvious fact that if each $S_k$ is a coordinate subspace,
then $\supp \lambda$ in (\ref{eqn:fractional_hitting_subpace}) can be assumed to consist of coordinate subspaces.
We point out as a consequence of this fact that the BL-polytope defined by coordinate projections
is equal to the perfect fractional set-packing polytope, 
which sharpens a result of Finner~\cite{Finner1992} on a generalization of H\"older's inequality.

\section{Preliminaries}\label{sec:preliminaries}

For positive integers $m,n$, let $[n] := \{1,2,\ldots,n\}$ and $[m,n] := \{m,m+1,\ldots,n\}$.
For a subset $J   \subseteq [n]$, 
let ${\bf 1}_{J} \in \RR^n$ denote the characteristic vector
defined by $({\bf 1}_{J})_i := 1$ if $i \in J$ and $({\bf 1}_{J})_i := 0$ otherwise.
If $J  = \{j_1,j_2,\ldots,j_k\}$, then ${\bf 1}_{J}$ is also written as 
${\bf 1}_{j_1,j_2,\ldots,j_k}$.
The all-one vector ${\bf 1}_{[n]}$ is simply denoted by ${\bf 1}$. The zero vector is denoted by ${\bf 0}$.
For two vectors $p,q \in \RR^{n}$, define $\min (p,q), \max (p,q) \in \RR^n$ 
by $(\min (p,q))_i := \min (p_i,q_i)$, $(\max (p,q))_i := \max (p_i,q_i)$ $(i \in [n])$. 
Let $\RR^n_{\downarrow}$ denote the set of vectors 
$p \in \RR^n$ with $p_1 \geq p_2 \geq \cdots \geq p_n$.

\subsection{Linear algebra}
For vectors $v_1,v_2,\ldots,v_k \in \KK^n$, 
let $\langle v_1,v_2,\ldots,v_k \rangle$ denote the vector subspace spanned by $v_1,v_2,\ldots,v_k$.
The following property ({\em supermodularity}) 
of the dimension function is elementary.
\begin{Lem}[{see \cite[Lemma 4.5]{HIOS2025}}]
\label{lem:supermodular}
Let $S,X,Y \in  {\cal S}(\KK^n)$ be vector subspaces.
\begin{itemize}
\item[(1)] $\dim S \cap X + \dim S \cap Y \leq \dim S \cap (X \cap Y) + \dim S \cap (X + Y)$.
\item[(2)] $(\dim X)^2 + (\dim Y)^2 < (\dim X \cap Y)^2 + (\dim X + Y)^2$ if $X \not \subseteq Y$ and  $Y \not \subseteq X$. 
\end{itemize}
\end{Lem}
\begin{proof}
(1). Use $\dim U+ \dim V = \dim U \cap V + \dim U+V$ and $\dim ((S \cap X) + (S \cap Y)) \leq \dim S \cap (X +  Y)$.
(2). There is a basis $v_1,v_2,\ldots,v_n$
such that $X =\langle v_i \mid  i \in I \rangle$ and $Y = \langle v_j \mid j \in J\rangle$ for some $I,J \subseteq [n]$ ; see Lemma~\ref{lem:apartment} below. Then 
$\dim X = |I|$, $\dim Y= |J|$, $\dim X \cap Y = |I \cap J|$, and $\dim X+Y = |I \cup J|$.
Then the statement follows 
from the strict convexity $x \mapsto x^{2}$ 
and $|I \cap J|< \min (|I|,|J|) \leq \max (|I|,|J|)  < |I\cup J|$.
\end{proof}

A {\em flag} is a family of vector subspaces $X_1,X_2,\ldots,X_m$ with $X_1 \subset X_2 \subset \cdots \subset X_m$.
If $m = n$, $\dim X_i = i$, and $X_n = \KK^n$, then
it is called a {\em complete flag}.
A complete flag is denoted by a calligraphic letter, such as ${\cal X}$, 
where the vector subspace with dimension $i$ in ${\cal X}$ is denoted by $X_i$.
The following property is also standard. 
\begin{Lem}[{See e.g., \cite[Part II, Lemma 10.80]{BridosonHaedfliger_Nonpositive}}]\label{lem:apartment}
For two complete flags ${\cal X}, {\cal Y}$, 
there is a basis $v_1,v_2,\ldots,v_n$ of $\KK^n$ 
such that every vector subspace in ${\cal X} \cup {\cal Y}$
is spanned by some subset of $v_1,v_2,\ldots,v_n$. 
\end{Lem}

We need an algorithm for constructing such a basis in the following special case.
\begin{Lem}\label{lem:n^3}
Let $v_1,v_2,\ldots,v_n$ be a basis of $\KK^n$ and let $s_1,s_2,\ldots,s_k$ 
be linearly independent vectors in $\KK^n$.
We can compute by $O(n^3)$ arithmetic operations of $\KK$ a basis $u_1,u_2,\ldots,u_n$ and indices $j_1,j_2,\ldots,j_k \in [n]$
such that
$\langle v_1,v_2,\ldots,v_i \rangle = \langle u_1,u_2,\ldots,u_i \rangle$ for $i \in [n]$ 
and $\langle s_1,s_2,\ldots,s_k \rangle = \langle u_{j_1},u_{j_2},\ldots u_{j_k}\rangle$.
\end{Lem}
\begin{proof}
The algorithm is based on Gaussian elimination.
Let $V := (v_1\ v_2\ \cdots \ v_n)$ and $S := (s_1\ s_2\ \cdots \ s_k)$, and  
consider the block matrix $(V\ S)$.
Multiply $V^{-1}$ from the left to obtain $(I\ V^{-1}S)$.
Next compute the column echelon form $L$ of $V^{-1}S$, 
where there are $1 \leq j_1 < j_2 < \cdots < j_k \leq n$ 
such that $L_{ih} = 0$ for $i > j_h$ and $h \in [k]$.  
Let $e_1,e_2,\ldots,e_n$ be the column vectors of $I$(the standard basis).
Define another basis $e'_1,e'_2,\ldots,e'_n$ 
by $e'_j := \ell_i$ if $j = j_i$ and $e'_j := e_j$ otherwise, 
where $\ell_i$ is the $i$-th column vector of $L$.
Since $(e'_1\ e'_2\ \cdots e'_n)$ is an upper triangular matrix 
with nonzero diagonals, it holds $\langle e_1,e_2,\ldots,e_i \rangle = \langle e'_1,e'_2,\ldots,e'_i \rangle$
for $i \in [n]$.
Also $\langle e'_{j_1},e'_{j_2},\ldots, e'_{j_k}\rangle = \langle V^{-1}s_j \mid j \in [k]\rangle$.
Then we obtain a desired basis $u_i := V e'_i$ $(i \in [n])$ and indices $j_1,j_2,\ldots,j_k$. 
This requires $O(n^3)$ arithmetic operations.
\end{proof}


\subsection{Hadamard spaces}
We recall the basic notions of Hadamard spaces and Busemann functions; see \cite{Bacak,BridosonHaedfliger_Nonpositive} for details.
In a metric space $({\cal H}, d)$, a {\em geodesic} is a path $\gamma:[a,b] \to {\cal H}$
such that $d(\gamma(s),\gamma(t)) = d(\gamma(a),\gamma(b))|s-t|/|a-b|$ holds for $s,t \in [a,b]$, 
where $s_\gamma := d(\gamma(a),\gamma(b))/|a-b|$  is called the {\em speed} of $\gamma$.
If $s_\gamma =1$, then $\gamma$ is called a unit-speed geodesic. 
Note that $t \mapsto \gamma(t/s_{\gamma})$ is a unit-speed geodesic for any geodesic $\gamma$.
The constant function $t \mapsto p$ is regarded as a geodesic with zero speed.

A {\em geodesic metric space} is a metric space $({\cal H}, d)$ 
such that every pair of points $x,y$ is connected by a geodesic.
A geodesic metric space $({\cal H},d)$ is called {\it CAT(0)} if for every $y \in {\cal H}$ and geodesic $\gamma:[a,b] \to {\cal H}$
the function $t \mapsto \frac{1}{2}d(y,\gamma(t))^2$
is $1$-strongly convex.
In a CAT(0) space, a geodesic connecting any pair of points is uniquely determined.
A complete CAT(0) space is particularly called a {\em Hadamard space}.

Let $({\cal H},d)$ be a Hadamard space. 
 A subset $K \subseteq {\cal H}$ is called  {\em convex} if
a geodesic connecting any pair of points in $K$ belongs to $K$.
Let $K \subseteq {\cal H}$ be a closed convex set.
For a point $x$ in ${\cal H}$,  
there is a unique point $x^*$ in $K$ 
having the minimum distance $\min_{y\in K} d(x,y)$ from $x$ to $K$.
This gives rise to the {\em projection} $\proj_K:{\cal H} \to K$
defined by $x \mapsto x^*$.

A function $f: {\cal H} \to \RR$ is called {\em convex} 
if for every geodesic $\gamma: [a,b] \to {\cal H}$
one-dimensional function $t \mapsto f(\gamma(t))$ is convex.
We recall an important class of convex functions, called {\em Busemann functions}, 
which are determined by geodesic rays.
By a {\em geodesic ray} we mean $\gamma:[0,\infty) \to {\cal H}$ 
such that the restriction $\gamma|_{[a,b]}$ for each $a,b \in [0,\infty)$ is a geodesic, where 
the speed of $\gamma$ is (well-)defined by $s_\gamma := d(\gamma(a),\gamma(b))/|a-b|$ $(a \neq b)$.
The Busemann function $b_{\gamma}:{\cal H} \to \RR$ associated with a geodesic ray $\gamma$ is defined by
\begin{equation}
b_{\gamma}(x) := s_{\gamma} \lim_{t \to \infty} d(x,\gamma(t/s_{\gamma})) - t \quad (x \in {\cal H}),
\end{equation}
where we let $b_{\gamma}(x) := 0$ if $s_\gamma = 0$.
Busemann functions are analogues of affine functions in Euclidean space.
To see this, suppose that ${\cal H} = \RR^n$ and $\gamma(t) = a + v t$ for 
$a,v \in \RR^n$ with $v \neq {\bf 0}$.
Letting $u := v/{\|v\|_2}$, we have
\begin{eqnarray}
&& b_{\gamma}(x) = \|v\|_2 \lim_{t \to \infty} \|x - a - ut  \|_2 - t = 
\|v\|_2 \lim_{t \to \infty} \frac{\|x - a - ut  \|^2_2 - t^2}{\|x - a - ut  \|_2 + t } \nonumber\\
&& = \|v\|_2 \lim_{t \to \infty} \frac{\|x-a\|_2^2- 2u^{\top}(x-a) t}{\|x - a - ut  \|_2 + t} = \|v\|_2 \lim_{t \to \infty} \frac{\|x-a\|_2^2/t- 2 u^{\top}(x-a) }{\|(x - a)/t - u  \|_2 + 1} \nonumber\\
&& = - \|v\|_2  u^{\top}(x-a) = - v^{\top}(x-a).\label{eqn:Euclidean_view}
\end{eqnarray}

\subsection{Horospherically convex optimization}
Here we briefly introduce horospherically convex optimization~\cite{CriscitielloKim2025,GoodwinLewisGenaroNicolae2026}.
Let $({\cal H},d)$ be a Hadamard space and $K \subseteq {\cal H}$ a convex set.
Let $f:K \to \RR$ be a real-valued function. 
A {\em Busemann subgradient} of $f$ at point $x \in K$ is
a geodesic ray $\gamma:[0,\infty) \to {\cal H}$ with $\gamma(0) = x$ satisfying
\begin{equation}\label{eqn:Busemann_subgradient}
f(y) \geq f(x) + b_\gamma(y)  \quad (y \in K).
\end{equation}
Note that \cite{GoodwinLewisGenaroNicolae2026} defines the Busemann subgradient as a pair $(s,\gamma)$ of $s \in \RR_{\geq 0}$ and unit-speed geodesic 
ray $\gamma$ issuing at $x$ such that $f(y) \geq f(x) + sb_\gamma(y)$ for all $y\in{\cal H}$.
The equivalence is clear by the correspondence $\gamma(t) \Leftrightarrow (s_\gamma, \gamma(t/s_\gamma))$.
Then $f:K \to \RR$ is called {\em horospherically convex}~\cite{CriscitielloKim2025}, 
or {\em Busemann subdifferentiable}~\cite{GoodwinLewisGenaroNicolae2026},
if it admits a Busemann subgradient at every point in $K$.

Suppose that 
${\cal H}$ has the {\em geodesic extension property}, that is,
every geodesic can be extended to a geodesic ray.
As in \cite{GoodwinLewisGenaroNicolae2026},
we consider the minimization of a sum of horospherically convex functions $f_0,f_1,\dots,f_{m}$ defined on a closed convex set $K \subseteq {\cal H}$ 
\begin{equation}
\mbox{Min.} \quad f(x) := \sum_{k=0}^{m} f_k(x)  \quad \mbox{s.t.} \quad x \in K.
\end{equation}
To apply the {\em incremental Busemann subgradient method} in \cite{GoodwinLewisGenaroNicolae2026}, we assume:
\begin{itemize}
\item $\inf_{x\in K} f(x) =: f^* > -\infty$, and there is $x^* \in K$ with $f(x^*) = f^*$.
\item The upper bound $D$ of the diameter of $K$ is known in advance.
\item An initial point $x^0 \in K$ is available.
\item A Busemann subgradient $\gamma$ of each $f_k$ at any $x\in K$ is available, 
where the speed $s_{\gamma}$ is bounded by some constant $L$ known in advance.
\item The projection $\proj_K:{\cal H} \to K$ is available.
\end{itemize}

\begin{description}
\item[{Incremental Busemann subgradient method}]\cite[Algorithm 2]{GoodwinLewisGenaroNicolae2026}
\item[0:] $x^0 \leftarrow$ any point in $K$, $i \leftarrow 0$.
\item[1:] $t_i \leftarrow$ step size.
\item[2:] For $k = 0,1,2,\ldots,m$, do
\begin{description}
\item[2-1:] $\gamma \leftarrow $ a Busemann subgradient of $f_k$ at $x^{i}$, where $s_\gamma \leq L$.
\item[2-2] $x^{i} \leftarrow \proj_K(\gamma(t_i))$.
\end{description}
\item[3:] $x^{i+1} \leftarrow x^{i}$, $i \leftarrow i+1$ and go to 1. 
\end{description}

Then, the following convergence result holds, 
where we state it in a slightly weaker form to allow rational step sizes.
\begin{Thm}[{\cite[Theorem 6.4]{GoodwinLewisGenaroNicolae2026}}]\label{thm:incremental_subgradient}
Suppose that the step size $t_i$ satisfies
\begin{equation}
\frac{D}{2L(m+1)\sqrt{i+1}} \leq t_i \leq \frac{D}{L(m+1)\sqrt{i+1}}.
\end{equation}
Then it holds
\begin{equation}
\min_{i=1,2,\ldots,N} f(x^i) - f^* \leq  
\frac{4(1+\log 3)(m+1)LD}{{\sqrt{N+2}}}.
\end{equation}
\end{Thm}
\begin{proof}
This follows from modifying the last step of their proof as
\begin{eqnarray*}
&& \min_{i=1,2,\ldots,N} f(x^i) - f^* \leq  \frac{D^2+L^2(m+1)^2 \sum_{j=\lceil N/2 \rceil}^N t_j^2}{2 \sum_{j=\lceil N/2 \rceil}^N t_j} \\
&& \leq (m+1)LD \frac{1+ \sum_{j=\lceil N/2 \rceil}^N 1/(j+1)}
{\sum_{j=\lceil N/2 \rceil}^N 1/\sqrt{j+1}} \leq \frac{4(1+\log 3)(m+1)LD}{{\sqrt{N+2}}}.
\end{eqnarray*}

\end{proof}

\subsection{Euclidean buildings}
{\em Euclidean buildings}~\cite{AbramenkoBrown_Building} constitute a representative class of Hadamard spaces without manifold structure. We will utilize special Euclidean buildings obtained 
from complete flags of $\KK^n$; for our purpose,
\cite[p.340--346]{BridosonHaedfliger_Nonpositive} is a useful reference.
The Euclidean building used in this paper, denoted by ${\cal B} = {\cal B}(\KK^n)$, is the set of all pairs $(p, {\cal X})$ of nonincreasing vectors $p \in \RR^{n}_{\downarrow}$ and complete flags ${\cal X}$ of $\KK^n$,
modulo the equivalence relation:
$(p,{\cal X}) \sim (q,{\cal Y})$ $\Leftrightarrow$ $p = q$ 
and $X_i = Y_i$ if $p_i>p_{i+1}$.
An equivalence class $[(p,{\cal X})]$ is also denoted by $(p,{\cal X})$.
The space ${\cal B}$ can be viewed as the product of $\RR$ 
and the Euclidean cone over the spherical building of type A; in particular, it is a Euclidean building.
For a fixed flag ${\cal X}$, 
the set of all $(p,{\cal X})$ for $p \in \RR^n_{\downarrow}$ is called a {\em chamber}.
Then ${\cal B}$ is 
obtained by gluing chambers along~$\sim$.
The distance on each chamber $\Delta$ is induced by the bijection $(p,{\cal X}) \mapsto p$ 
from $\Delta$ to a convex cone $\RR^{n}_{\downarrow}$ in Euclidean space $\RR^n$. 
Extending the distance~$d$ on each chamber geodesically, 
the building ${\cal B}$ becomes a geodesic metric space, and is actually a Hadamard space.

A point $(p,{\cal X})$ is uniquely represented by 
a formal sum 
of vector subspaces as
\begin{equation}\label{eqn:expression}
(p,{\cal X})\ \longleftrightarrow\  \sum_{i=1}^{n} (p_{i} - p_{i+1}) X_i,
\end{equation}
where $p_{n+1} := 0$, and the coefficient of $X_i$ is nonnegative except that of $X_n = \KK^n$.
In this expression, the zero vector space $\{{\bf 0}\}$ corresponds to the zero element.

An {\em apartment} ${\cal A} = {\cal A}(v_1,v_2,\ldots,v_n)$
associated with an ordered basis $(v_1,v_2,\ldots,v_n)$ of $\KK^n$ 
is the union of $n !$ chambers for flags of form 
\[
\{0\} < \langle v_{i_1} \rangle < \langle v_{i_1},v_{i_2} \rangle < \cdots < \langle v_{i_1},v_{i_2},\ldots,v_{i_n} \rangle = \KK^n,
\] 
where $(i_1,i_2,\ldots,i_n)$ is a permutation of $[n]$.
An apartment is an isometric subspace isometric to $\RR^n$. 
The isometry between $\RR^n$ and ${\cal A}$ is given by the following correspondence:
\begin{equation}\label{eqn:correspondence}
p = \sum_{k=1}^{n} (p_{j_k} - p_{j_{k+1}}) {\bf 1}_{ j_1,j_2,\ldots, j_k}\ \longleftrightarrow \ \sum_{k=1}^{n} (p_{j_k} - p_{j_{k+1}}) \langle v_{j_1},v_{j_2},\ldots, v_{j_k}\rangle,
\end{equation}
where $p \in \RR^n$ is ordered as $p_{j_1}\geq p_{j_2} \geq \cdots \geq p_{j_n}$, and $p_{j_{n+1}} :=0$.
By this correspondence~(\ref{eqn:correspondence}), an $n$-dimensional vector $p \in \RR^n$ uniquely 
specifies a point in ${\cal A}$, and is called the 
${\cal A}$-{\em coordinate} of the point.
By Lemma~\ref{lem:apartment}, for any two points in ${\cal B}$ there is an apartment containing them.
This property gives an explicit construction of a geodesic between any two points and also implies the geodesic extension property of ${\cal B}$. 

Next we consider the Busemann function $b_S = b_{\gamma_S}$ of the following geodesic ray 
$\gamma_S$ associated with a (nonzero) vector subspace $S \in {\cal S}(\KK^n)$:
\begin{equation}
\gamma_S(t):=  (t{\bf 1}_{[\dim S]},{\cal X}) = tS \quad (t \in [0,\infty)),
\end{equation}
where ${\cal X}$ is any complete flag with $X_{\dim S} = S$, 
and $tS$ is the right expression in (\ref{eqn:expression}). 
For $S = \KK^n$, we can define the geodesic of 
the negative direction:
\begin{equation}
\gamma_{-\KK^n} (t):= - t \KK^n \quad (= \gamma_{\KK^n}(-t)),
\end{equation}
We give an explicit formula of the Busemann function $b_{S} := b_{\gamma_S}$.
In the proof, we see that $\gamma_S$ is actually 
a geodesic ray with speed $\sqrt{\dim S}$.
\begin{Lem}[{Implicit in \cite{HiraiHadamard,KLM2009JDG}}]\label{lem:b_S}
For $S \in {\cal S}(\KK^n)$, it holds
\begin{equation}\label{eqn:Busemann_formula}
 b_{S}(p,{\cal X}) = - \sum_{i=1}^{n} p_{i} (\dim S \cap X_i - \dim S \cap X_{i-1}) \quad ((p,{\cal X})\in {\cal B}).
\end{equation}
\end{Lem}
\begin{proof}
Choose an apartment ${\cal A} = {\cal A}(v_1,v_2,\ldots,v_n)$ 
containing points $(p,{\cal X})$ and $S$. 
By Lemma~\ref{lem:apartment},
we may assume that $X_i = \langle v_1,\ldots,v_i\rangle$ $(i \in [n])$
and $S = \langle v_{j_1},\ldots,v_{j_k} \rangle$ where $k:= \dim S$.
In the ${\cal A}$-coordinate, the point $(p,{\cal X})$ is written as $p$, and  
the geodesic $\gamma_S$ is written as
\begin{equation}
\gamma_S(t) = t \langle v_{j_1},v_{j_2},\ldots,v_{j_k} \rangle = t {\bf 1}_{j_1,\ldots,j_k} \quad (t \in [0,\infty)).
\end{equation}
Therefore, $\gamma_S$ is a geodesic ray in ${\cal A} \simeq \RR^n$, and is a geodesic ray in ${\cal B}$. 
By~(\ref{eqn:Euclidean_view}), it holds
\begin{eqnarray*}
&& b_S(p,{\cal X}) =  - \sum_{\ell =1}^k p_{j_{\ell}} = -\sum_{i= 1}^{n} p_i (|[i] \cap \{j_1,\ldots, j_k\}| - |[i-1] \cap \{j_1,\ldots, j_k\}| ) \\
&& = - \sum_{i= 1}^{n} p_i (\dim X_i \cap S - \dim X_{i-1} \cap S). \qedhere
\end{eqnarray*}
\end{proof}

We next consider Busemann subgradients of $b_S$. 
\begin{Lem}
For $S \in  {\cal S}(\KK^n)$,
the Busemann subgradient $\gamma$ of $b_S$ at $(p,{\cal X})$ is given by
\begin{equation}
\gamma(t) = p + t {\bf 1}_{j_1,j_2,\ldots,j_k} \quad (t \in [0,\infty))
\end{equation}
where the geodesic is considered in the ${\cal A}$-coordinate of an apartment ${\cal A} = {\cal A}(v_1,v_2,\ldots,v_n)$ containing $p$ such that 
 $X_i = \langle v_1,\ldots,v_i\rangle$ $(i \in [n])$ and $S = \langle v_{j_1},\ldots,v_{j_k} \rangle$.
 Additionally,
the Busemann subgradient $\gamma$ of $b_{-\KK^n}$ at $(p,{\cal X})$ is given by $\gamma(t) = p - t{\bf 1}$.
\end{Lem}
\begin{proof}
Since $d(\gamma(t), \gamma_{S}(t)) = \|\gamma(t)-\gamma_S(t)\|^2_2 =\|p\|_2^2 < \infty  (\forall t \in [0,\infty))$, 
the two geodesic rays $\gamma$ and $\gamma_S$ are {\em asymptotic} (\cite[Definition II.8.1]{BridosonHaedfliger_Nonpositive}).
It is known \cite[Corollary II.8.20]{BridosonHaedfliger_Nonpositive} that the Busemann functions for two asymptotic geodesic rays
are equal up to some additive constant $C := b_S(q,{\cal Y}) - b_{\gamma}(q,{\cal Y})$ $(\forall (q,{\cal Y}) \in {\cal B})$.
The constant $C$ equals $b_{S}(p, {\cal X})$ since $b_\gamma(p,{\cal X}) = 0$. Then we have the defining property 
of the Busemann subgradient~(\ref{eqn:Busemann_subgradient}):
\begin{equation*}
b_S(q,{\cal Y}) = b_S(p,{\cal X}) + b_\gamma(q,{\cal Y}) \quad ((q,{\cal Y}) \in {\cal B}).
\end{equation*}
The latter statement follows from $b_{-\KK^n} = - b_{\KK^n}$.
\end{proof}

For $c \geq 0$, let $[{\bf 0}, c \KK^n]$ 
denote the subset of ${\cal B}$ consisting of 
points $(p,{\cal X})$ 
with $p_1 \leq c$ and $p_n \geq 0$.

\begin{Lem}\label{lem:[0,cK^n]}
$[{\bf 0}, c \KK^n]$ is a closed convex set with diameter $c\sqrt{n}$, 
where the projection is given by
\begin{equation}\label{eqn:proj}
{\textstyle \proj_{[{\bf 0}, c \KK^n]}} (p,{\cal X}) 
= (\min (c{\bf 1},\max({\bf 0},p)),{\cal X}) \quad ((p,{\cal X}) \in {\cal B}).
\end{equation}
\end{Lem}
\begin{proof}
For two points $x=(p,{\cal X}), y=(q,{\cal Y})$, there is an apartment ${\cal A} = {\cal A}(v_1,v_2,\ldots,v_n)$ containing $x,y$ (Lemma~\ref{lem:apartment}).
In the ${\cal A}$-coordinate, the intersection ${\cal A} \cap[{\bf 0}, c \KK^n]$ is given by:
\begin{equation}\label{eqn:interval}
{\cal A} \cap [{\bf 0}, c \KK^n] = \{p \in \RR^n \mid 0 \leq p_i \leq c\ (i \in [n]) \}.
\end{equation}
Therefore, if $x,y \in [{\bf 0}, c \KK^n]$, 
then the segment $t \mapsto t p + (1-t)q$ (in the ${\cal A}$-coordinate)  
must be the unique geodesic between $x,y$, which belongs to $[{\bf 0}, c \KK^n]$. Thus $[{\bf 0}, c \KK^n]$ is a convex set with diameter $c\sqrt{n}$.

Next we show the projection formula (\ref{eqn:proj}), 
from which the closedness of $[{\bf 0}, c \KK^n]$ follows.
Suppose that $x = (p,{\cal X}) \not \in [{\bf 0}, c \KK^n]$. Then 
$x$ and $y:=(\min (c{\bf 1},\max({\bf 0},p)),{\cal X})$ belong to the same chamber $\Delta$. 
From the expression (\ref{eqn:interval}), we see that
this $y$ is the nearest point of $x$ in ${\cal A} \cap [{\bf 0}, c \KK^n]$
for every apartment containing $\Delta$. 
Since such apartments cover the whole space ${\cal B}$, it must be 
$\proj_{[{\bf 0}, c \KK^n]}(x) = y$.
\end{proof}

\section{Formulation and algorithm}\label{sec:formulation&algorithm}
\subsection{Formulation}

Let $S_1,S_2,\ldots,S_m$ 
be vector subspaces of $\KK^n$. 
For applications in Section~\ref{sec:application}, 
we consider the following generalization of 
the fractional subspace packing problem~(\ref{eqn:fractional_subspace_packing}).
Observe that the inequality system in (\ref{eqn:fractional_subspace_packing}) includes 
an upper bound condition $x_k \leq 1$, which is implied by $\sum_{k=1}^m x_k \dim  (S_k \cap X) \leq \dim X$ for $X = S_k$.
So we allow a possibly stricter rational upper bound $u \in ([0,1] \cap \QQ)^m$. 
Further we consider a nonnegative integer weight $c\in \ZZ_{\geq 0}^m$, where $c_{\rm max} := \max_{k \in [m]} c_k$. 
Our target LP is the following: 
\begin{eqnarray}
\mbox{Max.} && \sum_{k=1}^m c_k x_k \nonumber \\
\mbox{s.t.} && \sum_{k=1}^m  x_k \dim (S_k \cap X)\leq \dim X \quad (X \in  {\cal S}(\KK^n)), \nonumber \\
&& 0 \leq x_k \leq u_k \quad (k \in [m]). \label{eqn:weighted}
\end{eqnarray}
The dual LP is given by
\begin{eqnarray}
\mbox{Min.} && \sum_{X \in {\cal S}(\KK^n)}  \lambda(X) \dim X  + \sum_{k=1}^m u_k y_k  \nonumber \\
\mbox{s.t.} && \sum_{X \in {\cal S}(\KK^n)} \dim (S_k \cap X) \lambda(X) + y_k \geq c_k  \quad (k \in [m]), \nonumber \\
&& \lambda:{\cal S}(\KK^n) \to  \RR_{\geq 0}, |\supp \lambda| < \infty,\ y\in \RR_{\geq 0}^m.   \label{eqn:dual_weighted}
\end{eqnarray}
Then strong duality holds for the same reason as before.
We denote the optimal value of these LP by $\tau^*_{c,u}$.

We show that the dual LP~(\ref{eqn:dual_weighted}) 
can be formulated as a horospherically convex optimization on 
$[{\bf 0}, c_{\rm max}\KK^n]$.
The starting point is the following property, which is classically known for a general setting of lattice polyhedra 
by Gr\"oflin and Hoffman~\cite{GroflinHoffman1982}.
\begin{Lem}[{\cite{GroflinHoffman1982}; see \cite[p.70--71]{CLV01b} and \cite[p.979]{HIOS2025}}]
    There is an optimal solution $(\lambda,y)$ in (\ref{eqn:dual_weighted}) 
    such that $\supp \lambda$ is a flag.
\end{Lem}
The proof is the standard uncrossing argument.
\begin{proof} 
Choose an optimal solution $\lambda$ 
such that $\sum_{X} \lambda(X) (\dim X)^2$ is maximum\footnote{
Formally, consider the equivalence relation on vector subspaces: 
$X \equiv Y$ if $X$ and $Y$ have the same coefficient vectors (\ref{eqn:coefficient_vector}).
Regard $\lambda$ as a function on the equivalence class, and 
regard LP (\ref{eqn:dual_weighted}) as a finite-dimensional LP.}.
Suppose that there are $X,Y \in \supp \lambda$ with 
$X \not \subseteq Y$, $Y \not \subseteq X$.
Let $\epsilon = \min \{\lambda(X),\lambda(Y)\} > 0$.
Define $\lambda'$ by $\lambda'(X) := \lambda(X)-\epsilon$,  $\lambda'(Y) := \lambda(Y)-\epsilon$, 
$\lambda'(X\cap Y) := \lambda(X \cap Y) + \epsilon$, $\lambda'(X+ Y) := \lambda(X + Y) + \epsilon$, and $\lambda'(Z) := \lambda(Z)$ for other $Z$. 
By Lemma~\ref{lem:supermodular}~(1), the resulting $(\lambda',y)$ is feasible. Also the objective value does not change. Hence, $\lambda'$ is also optimal.
However, Lemma~\ref{lem:supermodular}~(2) implies $\sum_{X} \lambda(X) (\dim X)^2< \sum_{X} \lambda'(X) (\dim X)^2$, which contradicts the choice of $\lambda$.
\end{proof}
Then, (\ref{eqn:dual_weighted}) becomes
\begin{eqnarray}
\mbox{Min.} && \sum_{i=1}^n \lambda_i \dim X_i 
+ \sum_{k=1}^m u_k y_k \nonumber \\
\mbox{s.t.} 
&& \sum_{i=1}^n \lambda_i \dim (S_k \cap X_i)  + y_k \geq c_k \quad (k \in [m]), \nonumber \\
&& {\cal X}: \mbox{complete flag in $\KK^n$},\nonumber \\
&& \lambda \in \RR^{n}_{\geq 0},\  y \in \RR_{\geq 0}^m. \label{eqn:building_LP_lambda} 
\end{eqnarray}
Since $\dim X_i = i$, it holds $\sum_{i=1}^n \lambda_i \dim X_i = \lambda_1 + 2 \lambda_2 + 3 \lambda_3 + \cdots + n \lambda_n$.
Define nonnegative nonincreasing vector $p \in \RR_{\downarrow}^{n}$ by
\begin{equation}
p_i := \lambda_i + \lambda_{i+1} + \cdots + \lambda_{n} \quad (i \in [n]).
\end{equation}
In (\ref{eqn:building_LP_lambda}), 
change the variable $\lambda$ to $p$, where 
$\lambda$ is recovered by 
\begin{equation}\label{eqn:recover}
\lambda_i = p_{i} - p_{i+1} \quad (i \in [n]), \quad p_{n+1} := 0.
\end{equation}
Observe the relation
\begin{equation}
\sum_{i=1}^n (p_i - p_{i+1}) \dim (S_k \cap X_i) =  \sum_{i=1}^n p_i \{\dim (S_k \cap X_i) -\dim (S_k \cap X_{i-1})\},
\end{equation}
where $X_0 = \{{\bf 0}\}$.
Then (\ref{eqn:building_LP_lambda}) becomes
\begin{eqnarray}
\mbox{Min.} && \sum_{i=1}^n p_i + \sum_{k=1}^m u_k y_k \nonumber \\
\mbox{s.t.} 
&& \sum_{i=1}^n p_i \{ \dim (S_k \cap X_i) - \dim (S_k \cap X_{i-1})\} +y_k  \geq c_k \quad (k \in [m]),\nonumber \\
&& {\cal X}: \mbox{complete flag in $\KK^n$}, \nonumber \\
&& p_1 \geq p_2 \geq \cdots \geq p_n \geq 0,\   y\in \RR_{\geq 0}^m.  \label{eqn:building_LP}
\end{eqnarray}
From $\dim (S_k \cap X_i) -\dim (S_k \cap X_{i-1}) \in \{0,1\}$, we see
\begin{Lem}
If $((p,{\cal X}),y)$ is feasible to (\ref{eqn:building_LP}), 
then $((\min (c_{\rm max}{\bf 1}, p), {\cal X}),y)$ is also feasible and does not increase the objective value. 
\end{Lem}
Therefore, we may add the constraint $p_1 \leq c_{\rm max}$ to (\ref{eqn:building_LP}).
By using the formula (\ref{eqn:Busemann_formula}) of Busemann functions, 
LP~(\ref{eqn:building_LP}) is written as 
\begin{eqnarray}
\mbox{Min.} && b_{-\KK^n}(p,{\cal X}) + \sum_{k=1}^m u_k y_k  \nonumber \\
\mbox{s.t.} 
&& - b_{S_k}(p,{\cal X}) + y_k \geq c_k, \quad (k \in [m]),\nonumber \\
&& (p,{\cal X}) \in [{\bf 0}, c_{\rm max} \KK^n],\ y \in \RR^{m}_{\geq 0}. \label{eqn:Busemann_LP0}
\end{eqnarray}
We eliminate $y$ to obtain: 
\begin{Thm}
$\tau^*_{c,u}$ is equal to the optimal value of
\begin{eqnarray}
\mbox{Min.} && b_{-\KK^n}(p,{\cal X}) + \sum_{k=1}^m u_k \max \{0,  c_k +  b_{S_k}(p,{\cal X})\}  \nonumber \\
\mbox{s.t.} 
&& (p,{\cal X}) \in [{\bf 0}, c_{\rm max} \KK^n]. \label{eqn:Busemann_LP}
\end{eqnarray}
\end{Thm}
In the next subsection, 
we apply the incremental subgradient method 
for this horospherically convex optimization problem.
We end this subsection by giving a procedure to 
obtain a solution $(p,{\cal X})$ with zero penalty
for the uncapacitated case~$u = {\bf 1}$.
\begin{Lem}\label{lem:recover}
Suppose that $u={\bf 1}$.
From a feasible solution $(p,{\cal X})$ in (\ref{eqn:Busemann_LP}),  
we can compute in $O(mn^3)$ time 
 a feasible solution $((q,{\cal Y}),{\bf 0})$ in (\ref{eqn:Busemann_LP0})
with
\begin{equation}\label{eqn:modify}
\sum_{i=1}^n q_i  \leq \sum_{i=1}^n p_i + \sum_{k=1}^m \max \{0,  c_k +  b_{S_k}(p,{\cal X})\}.  
\end{equation}
\end{Lem}
\begin{proof}
By $\dim S_k \cap X_i - \dim S_k \cap X_{i-1} \in \{0,1\}$, 
if $p_{i-1} > p_i$, 
then increasing $p_i$ by small $\epsilon > 0$ does not increase $c_k + b_{S_k}(q, {\cal Y})$.
In addition, if $\dim S_k \cap X_i - \dim S_k \cap X_{i-1} = 1$, 
then $c_k + b_{S_k}(q, {\cal Y})$ decreases by $\epsilon$

Consider $k$ such that $\alpha:= c_k + b_{S_k}(p, {\cal X}) > 0$.
By Lemma~\ref{lem:n^3},
choose a basis $u_1,u_2,\ldots,u_n$ and indices $j_1,j_2,\ldots,j_{\dim S_k}$ 
such that  $X_i = \langle u_1,u_2,\ldots,u_i\rangle$ $(i \in [n])$ and 
$S_k = \langle u_{j_1},u_{j_2},\ldots, u_{j_{\dim S_k}}\rangle$.
In the apartment ${\cal A} = {\cal A}(u_1,u_2,\ldots,u_n)$, 
the Busemann function for $S_k$ 
is written as $p \mapsto - \sum_{i=1}^{\dim S_k} p_{j_i}$.
Therefore, increase $p_{j_1}$ by $\alpha$, which does not increase the RHS of (\ref{eqn:modify}). 
Recover the expression $(p,{\cal X})$ according to (\ref{eqn:correspondence}).
Repeat this procedure until such $k$ does not exist. 
Then the second term of the RHS of (\ref{eqn:modify}) is zero.
Therefore, $((p,{\cal X}),{\bf 0})$ is feasible to (\ref{eqn:building_LP}). 
\end{proof}

\subsection{Algorithm}
Define $f_k:{\cal B} \to \RR$ $(k=0,1,\ldots,m)$ by 
\begin{eqnarray}
f_0(p,{\cal X}) &:= & b_{-\KK^n}(p,{\cal X}) = - b_{\KK^n}(p,{\cal X}),\\ f_{k}(p,{\cal X}) &:=& u_k \max\{0, c_k +  b_{S_k}(p,{\cal X})  \}
\quad (k=1,\ldots,m).
\end{eqnarray}
They are horospherically convex functions on ${\cal B}$, where the Busemann subgradient $\gamma$ of $f_k$ at $(p,{\cal X})$ is given by
\begin{equation}\label{eqn:subgradient}
\gamma(t) = 
\left\{ \begin{array}{ll} \displaystyle
p-  t{\bf 1}   & {\rm if}\ k =0,\\ 
p+  t u_k{\bf 1}_{j_1,j_2,\ldots,j_{\dim S_k}} & {\rm if}\ k \geq 1,  c_k + b_{S_k}(p,{\cal X}) >0, \\ 
p & {\rm otherwise},
\end{array} \right.
\end{equation}
where $\gamma$ is considered in the ${\cal A}$-coordinate
for an apartment ${\cal A} = {\cal A}(v_1,v_2,\ldots,v_n)$ such  
that $X_i = \langle v_1,\ldots,v_i\rangle$ $(i \in [n])$ 
and $S_k = \langle v_{j_1},\ldots,v_{j_{\dim S_k}} \rangle$ (if $k \geq 1$).

Now, we consider the following horospherically convex optimization
\begin{equation}
\mbox{Min.} \quad f(p,{\cal X}) := \sum_{k=0}^{m} f_k(p,{\cal X}) \quad \mbox{s.t.} \quad  (p,{\cal X}) \in [{\bf 0}, c^* \KK^n],
\end{equation}
and apply the incremental Busemann subgradient method. 
The initial point $x^{0}$ is chosen, say, 
$
x^{0} = ({\bf 0},{\cal X}),
$
where ${\cal X}$ is any complete flag. 
From (\ref{eqn:subgradient}) 
we can choose $L > 0$ as $L = \sqrt{n}$ (since $0 \leq u_k \leq 1$).
The diameter $D$ of $[{\bf 0}, c_{\rm max}\KK^n]$ 
is given by $c_{\rm max} \sqrt{n}$ (Lemma~\ref{lem:[0,cK^n]}).
We choose a rational step size $t_i$ by
$t_i := D/(L(m+1) \lceil \sqrt{i+1} \rceil) = c_{\rm max}/((m+1) \lceil \sqrt{i+1} \rceil)$.
Throughout the algorithm, 
we maintain a basis $v_1,v_2,\ldots,v_n$ of $\KK^n$ and a vector $p \in \RR^{n}_{\downarrow}$ 
to represent points $x^i$ as $x^{i} = (p,{\cal X})$ with $X_j = \langle v_1,v_2,\ldots,v_j\rangle$ $(j \in [n])$.
\begin{description}
\item[Incremental Busemann subgradient method for (\ref{eqn:Busemann_LP}):]
\item[0:] $x^{0}= (p,{\cal X}) \leftarrow ({\bf 0},{\cal X})$, $i \leftarrow 0$.
\item[1:] $t_i \leftarrow c_{\rm max} /((m+1)\lceil \sqrt{i+1} \rceil)$.  
\item[2:] For $k=0,1,2,\ldots,m$, do 
\begin{description}
\item[2-0:] If $k=0$, then $p \leftarrow  \max ({\bf 0}, p - t_i {\bf 1} )$, and go to {\bf 2-5}.
 \item[2-1:] Compute a basis $v_1,v_2,\ldots,v_n$ and indices $j_1,j_2,\ldots,j_{\dim S_{k}}$
such that $X_i = \langle v_1,v_2,\ldots,v_i\rangle$ $(i \in [n])$ and 
$S_k = \langle v_{j_1},v_{j_2},\ldots, v_{j_{\dim S_{k}} }\rangle$.
\item[2-2:] If $c_k +  b_{S_k}(p,{\cal X}) = c_k -p^{\top}{\bf 1}_{j_1,j_2,\ldots,j_{\dim S_{k}}} \leq 0$, then go to {\bf 2-5}.
\item[2-3:] $p \leftarrow \min (c_{\rm max}{\bf 1}, p + t_iu_k{\bf 1}_{j_1,j_2,\ldots,j_{\dim S_{k}}})$.
\item[2-4:] Sort $p$ as $p_{i_1} \geq p_{i_2} \geq \cdots \geq p_{i_n}$, and update $(p,{\cal X})$ by $p_j \leftarrow p_{i_j}$ and $X_j \leftarrow \langle v_{i_1},v_{i_2},\ldots,v_{i_j} \rangle$ for $j \in [n]$.
\item[2-5]: $x^{i} \leftarrow (p,{\cal X})$.
\end{description}
\item[3:] $x^{i+1} \leftarrow x^i$, $i\leftarrow i+1$, 
and go to {\bf 1}.
\end{description}

In step {\bf 1}, we compute $\lceil \sqrt{i+1} \rceil$ from $\lceil \sqrt{i} \rceil$ by
$\lceil \sqrt{i+1} \rceil = 1 + \lceil \sqrt{i} \rceil$ 
if $i = \lceil \sqrt{i} \rceil^2$ and 
$\lceil \sqrt{i+1} \rceil = \lceil \sqrt{i} \rceil$ otherwise. 
In steps {\bf 2-0} and {\bf 2-3}, the bit-length of each component of $p$ 
increases at most the bit-length of $t_i$.
By Lemma~\ref{lem:n^3},
step {\bf 2-1} requires $O(n^3)$ arithmetic operations in $\KK$.
By Theorem~\ref{thm:incremental_subgradient}, 
we have the following:
\begin{Thm}\label{thm:main_cu}
For $\epsilon > 0$, 
the incremental Busemann subgradient method finds a solution $(p,{\cal X}) \in [{\bf 0}, c_{\rm max} \KK^n]$ with
\[
f(p,{\cal X}) -  \tau^*_{c,u} \leq \epsilon
\]
in  $O(c_{\rm max}^2 n^2 m^2 \epsilon^{-2})$ iterations, 
where each iteration requires $O(n^3m)$ arithmetic operations in $\KK$  and $O(n m \log n)$ arithmetic operations in $\QQ$.
The necessary bit-length for obtaining $p$ is bounded 
by a polynomial in $c_{\rm max},n,m,\epsilon^{-1}$, and the bit-length of $u$.
\end{Thm}
Theorem~\ref{thm:main} is obtained from this theorem, Lemma~\ref{lem:recover}, and (\ref{eqn:recover}).

\paragraph{Computing $\tau^*_{c,{\bf 1}}$ exactly.}
We show that Theorem~\ref{thm:main_cu} yields an FPT-algorithm for 
computing $\tau^*$, parametrized by the dimension $n$.
\begin{Lem}
$\tau_{c,{\bf 1}}^*$ is a rational number $a/b$ 
with $b \leq  n^{n/2}$.
\end{Lem}
\begin{proof}
Let $((p,{\cal X}),y)$ be an optimal solution of (\ref{eqn:Busemann_LP}).
By Lemma~\ref{lem:recover}, we may assume $y ={\bf 0}$.
Then $\tau_{c,{\bf 1}}^* = \sum_{i=1}^n p_i$.
Observe that $p$ is also an optimal solution of the LP with ${\cal X}$ fixed.
We may further assume that $p$ is an extreme optimal solution.
Since
$\dim S_k \cap X_i- \dim S_k \cap X_{i-1} \in \{0,1\}$,
this LP is a combinatorial LP, that is, the coefficient matrix consists of $0,1,-1$. Since $p$ is a unique solution 
of the linear equations obtained from this inequality system, the common denominator of $p_1,p_2,\ldots,p_n$ is bounded 
by the absolute value of the determinant of an $0,1,-1$ matrix of size at most $n$.
By Hadamard's inequality, it is at most $n^{n/2}$, from which the result follows.
\end{proof}
Thus, by the continued fraction method \cite[Corollary 6.3a]{SchrijverIP},
we can determine $\tau_{c,{\bf 1}}^*$ if $f(p,{\cal X})-\tau_{c,{\bf 1}}^* \leq 1/(2 n^{n})$.
Such $(p,{\cal X})$ is obtained by the incremental subgradient method 
with accuracy $\epsilon <  1/(2 n^{n})$. 
\begin{Thm}\label{thm:fpt_c}
We can compute $\tau_{c,{\bf 1}}^*$ in $O(c_{\rm max}^2 n^{2n+5}m^3)$ time.
\end{Thm}





\section{Applications}\label{sec:application}

\subsection{Fractional linear matroid parity}
The fractional linear matroid parity problem is the case where $\dim S_k = 2$ for all $k \in [m]$.
In this case, it is known~\cite{GP13,
VandeVate92} that LP~(\ref{eqn:weighted}) and its dual~(\ref{eqn:dual_weighted}) have half-integrality. 
For simplicity, we consider the unweighted and the uncapacitated case $c=u = {\bf 1}$.
A solution $(p,{\cal X})$ in (\ref{eqn:Busemann_LP}) is simply called  {\em feasible} 
if $((p,{\cal X}),{\bf 0})$ is feasible to (\ref{eqn:Busemann_LP0}).
Then, a half-integral feasible solution $(p,{\cal X})$ is of the form
\begin{equation}
(p,{\cal X}) = \frac{1}{2} X + \frac{1}{2} Y 
\end{equation}
for some $X,Y \in {\cal S}(\KK^n)$ with $X \subseteq Y$, where we use the right expression of (\ref{eqn:expression}). 
The feasibility condition in (\ref{eqn:Busemann_LP0}) is written as 
\begin{equation}
\frac{1}{2} \dim X \cap S_k + \frac{1}{2} \dim Y \cap S_k \geq 1 \quad (k \in [m]).
\end{equation}
Such a pair $(X,Y)$  
is nothing but a (nested) {\em 2-cover}~in~\cite{CLV01a,CLV01b}. Thus, we have:
\begin{Thm}[\cite{CLV01b}]
$\tau^*$ is equal to 
the minimum of $(\dim X+ \dim Y)/2$ 
over all $2$-covers~$(X,Y)$.
\end{Thm}
We give a procedure to construct a convex combination of $2$-covers from any feasible solution $(p,{\cal X})$, which also proves the above theorem. 
\begin{Lem}
For a feasible solution $(p,{\cal X})$, decompose $p$ as
\begin{equation}\label{eqn:representation}
p = \frac{1}{2} {\bf 1} + \sum_{i=1}^{\ell+1} \rho_i \frac{{\bf 1}_{[\alpha_i]} - {\bf 1}_{[\beta_i+1,n]}}{2} 
= \sum_{i=1}^{\ell+1} \rho_i \frac{{\bf 1}_{[\alpha_i]} + {\bf 1}_{[\beta_i]}}{2}, 
\end{equation}
where $0 = \alpha_{\ell+1} \leq \alpha_{\ell} \leq \cdots \leq \alpha_2 \leq \alpha_{1} 
< \beta_{1} \leq \beta_{2} \leq \cdots \leq \beta_\ell \leq \beta_{\ell+1} = n$, and $\rho_i \geq 0$ $(i \in [\ell+1])$ with $\sum_{i=1}^{\ell+1} \rho_i =1$.
Then it holds
\begin{equation}
(p,{\cal X}) = \sum_{i=0}^{\ell+1} \rho_i \frac{X_{\alpha_i} +  X_{\beta_i}}{2},
\end{equation}
where each $(X_{\alpha_i},X_{\beta_i})$ is a $2$-cover.
\end{Lem}
The representation (\ref{eqn:representation}) is constructed easily: 
Let $\kappa \leftarrow p_1-p_n$, $q \leftarrow p -(\kappa/2){\bf 1}$, $i \leftarrow 1$.
Let $\alpha_i \leftarrow $ the maximum index $j$ with $(q^i)_j>0$, $\beta_i \leftarrow $ the maximum index $j$ with $(q^i)_j \geq 0$, 
let $\rho_i \leftarrow 2 \min (q_{\alpha_i}, -q_{\beta_i+1})$,
and let $q \leftarrow q - \rho_i ({\bf 1}_{[\alpha_i]}-{\bf 1}_{[\beta_i+1,n]})/2$, $i \leftarrow i+1$. 
Repeat it until $q = {\bf 0}$.
After that, let $\rho_{i} \leftarrow 1 - \kappa$, where $i = \ell +1$.
\begin{proof}
Suppose, for a contradiction, that  $(X_{\alpha_i},X_{\beta_i})$ is not a $2$-cover. There is $k \in [m]$ such that
\begin{equation}\label{eqn:violated}
\frac{1}{2}\dim X_{\alpha_i} \cap S_k + \frac{1}{2}\dim X_{\beta_i} \cap S_k <  1.
\end{equation}
That is, $\dim X_{\alpha_i} \cap S_k = 0$ and $\dim X_{\beta_i} \cap S_k \leq 1$.
Define indices $\alpha_*, \beta_*$ with $\alpha_* < \beta_*$ 
by $\{\alpha_*,\beta_*\} = \{i \in [n] \mid \dim S_k \cap X_i - \dim S_k \cap X_{i-1} =1\}$. 
Then, it holds $\alpha_i < \alpha_*$ and $\beta_{i} < \beta_*$.
For other indices $j (\neq i)$ it must hold
\begin{equation}\label{eqn:other}
\frac{1}{2}\dim X_{\alpha_{j}} \cap S_k + \frac{1}{2}\dim X_{\beta_{j}} \cap S_k \leq 1.
\end{equation}
Define $d^{(k)} \in \{0,1\}^n$ by $(d^{(k)})_i := \dim S_k \cap X_i - \dim S_k \cap X_{i-1}$.
Then we obtain a contradiction to the feasibility:
\begin{equation*}
1 \leq p^\top d^{(k)} 
= \sum_{i=1}^{\ell+1} \rho_i \left(\frac{{\bf 1}_{[\alpha_i]} + {\bf 1}_{[\beta_i]}}{2}\right)^{\top} d^{(k)} \nonumber \\
 =  \sum_{i=1}^{\ell+1} \rho_i \frac{\dim X_{\alpha_i} \cap S_k + \dim X_{\beta_i} \cap S_k}{2} < 1.   \qedhere
\end{equation*}
\end{proof}
Therefore,  
by the incremental Busemann subgradient method with $\epsilon < 1/2$ and 
$u = {\bf 1}$, 
we obtain a solution $(p,{\cal X})$ with
$
\sum_{i=1}^n p_i - \tau^* < 1/2.
$
Decompose $p$ as in (\ref{eqn:representation}), 
and obtain 2-covers $(X_{\alpha_i}, X_{\beta_i})$, one of which is an optimal 2-cover. 
Thus, we can obtain an optimal 2-cover and optimal value $\tau^*$ in $O(n^5m^3)$ time.
For $k=1,2,\ldots,m$, 
apply (again) the incremental Busemann subgradient method with $u_k = 0,1/2$.
Fix $u_k$ as the minimum of $u_k = 0,1/2,1$ so that the optimal value does not decrease. 
Then the resulting $u$ is a half-integral optimal fractional matching. 
The whole procedure is done
in $O(n^5m^4)$ time.
This is slower than the CLV-algorithm~\cite{CLV01a} that requires $m^2$ calls of a linear matroid intersection algorithm; the current fastest one~\cite{Harvey2009Algebraic} has time complexity $O(n^{\omega-1}m)$, where $\omega$ denotes the matrix multiplication exponent.
An advantage of our algorithm is its simplicity, 
whereas the CLV-algorithm is complicated and needs a rather involved proof of correctness.
We hope that further development 
on horospherically convex optimization 
will lead to a better algorithm.

\subsection{Brascamp-Lieb polytopes}

Our motivation for studying the fractional subspace packing problem~(\ref{eqn:fractional_subspace_packing}) comes from the {\em Brascamp-Lieb inequality}~\cite{BrascampLieb,Lieb}.
Given surjective linear maps $B_k:\RR^n \to\RR^{r_k}$ and nonnegative reals 
$u_k \in \RR_{\geq 0}$ for $k\in [m]$ ({\em BL-datum}), 
the Brascamp-Lieb inequality says that 
for some constant $C \in \RR_{\geq 0} \cup \{+\infty\}$ it holds
\begin{equation}\label{eqn:BL}
\int_{x \in \RR^n} \prod_{k=1}^m f_k(B_kx)^{u_k}
{\rm d}x
\leq C \prod_{k=1}^m \left( \int_{x_k \in \RR^{r_k}} f_k(x_k) {\rm d}x_k \right)^{u_k} 
\end{equation}
for any nonnegative-valued measurable functions $f_k: \RR^{r_k} \to \RR_{\geq 0}$ $(k\in [m])$.
The minimum possible $C$, denoted by $C_{\rm BL}$, is called 
the {\em BL-constant}.
The interesting case of the BL-inequality is 
when $C_{\rm BL} < \infty$, and is characterized by Bennett et al.~\cite{BCCT08} as follows:
\begin{Thm}[\cite{BCCT08}]\label{thm:BCCT}
Given BL-datum $B_k:\RR^n \to  \RR^{r_k}$, $u_k \in \RR_{\geq 0}$ $(k\in [m])$,
the BL-constant is finite if and only if 
it holds 
 $n = \sum_{k=1}^{m}r_ku_k$  and
$\sum_{k=1}^m u_k \dim B_k V \geq \dim V$ 
for every $V \in  {\cal S}(\RR^n)$.
\end{Thm}
The {\em BL-polytope} ${\rm BLP}({\cal B})$ associated with ${\cal B} = \{B_k\}_{k \in [m]}$ is defined as a rational polytope of
all vectors $u = (u_k) \in \RR^m_{\geq 0}$ satisfying the condition in Theorem~\ref{thm:BCCT}.
We suppose that each $B_k$ is given as an $r_k \times n$ matrix with full row rank.
The computational complexity of BL-polytopes is an intriguing open problem in theoretical computer science.
It is not known whether the membership problem for BL-polytopes is NP-hard or polynomial-time solvable.
The current best result is due to Garg et al.~\cite{GGOW18} who gave a polynomial-time membership algorithm 
for input rational vectors $u$ having a bounded common denominator of entries.
For the special case of rank-2 matrices $B_k$, 
Hirai et al.~\cite{HIOS2025} showed a polynomial-time membership algorithm by utilizing the connection 
to the fractional linear matroid parity.
As suggested in \cite[Lemma 7.1]{FranksSomaGoemans2022} (see also \cite[Proposition 4.3]{BezGauvanTsuji2025}),   
the inequalities  $\sum_{k=1}^m u_k \dim B_k V \geq \dim V$ 
are linked to the subspace packing condition 
$\sum_{k=1}^m u_k \dim S_k \cap X \leq \dim X$ as follows:
For $k \in [m]$, define subspaces $S_k \in {\cal S}(\RR^n)$ by 
\begin{equation}
S_k := \image B_k^{\top},
\end{equation}
Then $\dim S_k = r_k$.
Observe that $u \in (B_k V)^{\bot}$ if and only if $B_k^{\top} u \in V^{\bot}$, where $(\cdot)^{\bot}$ denotes 
the orthogonal complement with respect to the standard inner product.
Since $u \mapsto B_k^{\top}u$ is injective, it holds 
$\dim (B_k V)^{\bot} = \dim \{u \in V \mid B_k^{\top}u \in V^{\bot}\} = \dim \image B_k^{\top} \cap V^{\bot} = \dim S_k \cap V^{\bot}$.
With $X := V^{\bot}$ and $\dim B_k V = r_k - \dim S_k \cap X$, we have
\begin{eqnarray}
&& \sum_{k=1}^m u_k \dim B_k V - \dim V  \nonumber \\
&& =  - \sum_{k =1}^m u_k \dim S_k \cap X + \dim X - \left( n - \sum_{k=1}^m r_k u_k \right). \label{eqn:VXrelation}
\end{eqnarray}
The third term of RHS is zero if and only if $\dim \RR^n = \sum_{k=1}^{m}u_k\dim S_k \cap \RR^n$.
Thus, the BL-polytope is the face of the feasible region of (\ref{eqn:fractional_subspace_packing}).
In particular, it is
viewed as the polytope of ``perfect" fractional subspace packings.

To formulate approximate membership, 
we consider a relaxation of the BL-polytope.
For $\epsilon \geq 0$, the {\em $\epsilon$-relaxed BL-polytope} ${\rm BLP}_{\epsilon}({\cal B})$
is defined as the set of all vectors $u$ satisfying
$n \geq \sum_{k=1}^{m}u_kr_k \geq n-\epsilon$ 
and  $\sum_{k=1}^m u_k \dim B_k V \geq \dim V-\epsilon$ for every $V \in {\cal S}(\RR^n)$.
Consider LP~(\ref{eqn:dual_weighted}) with weight $r = (r_k)$ and capacity $u$.
From the relation~(\ref{eqn:VXrelation}), we have:
\begin{Lem}\label{lem:BLP_epsilon}
For $\epsilon \geq 0$ and $u \in [0,1]^m$, if $\tau_{r,u}^* \geq n- \epsilon$, then  ${\rm BLP}_{\epsilon}({\cal B}) \cap [0,u] \neq \emptyset$. 
\end{Lem}

Consider a rational vector $u \in \QQ^m$ satisfying necessary conditions $\sum_{k=1}^m r_k u_k = n$ and $0 \leq u_k \leq 1$ $(k \in [m])$ for 
belonging to ${\rm BLP}({\cal B})$. 
Apply the incremental Busemann subgradient method for (\ref{eqn:Busemann_LP}) with accuracy $\epsilon > 0$.
If the obtained solution $(p,{\cal X})$ with $f(p,{\cal X}) \leq \tau^*_{r,u} + \epsilon$ also satisfies $n \leq f(p,{\cal X})$,
then there is $x \in {\rm BLP}_{\epsilon}({\cal B})$ with $x \leq u$, and necessarily $u \in {\rm BLP}_{\epsilon}({\cal B})$.
Otherwise $n >  f(p,{\cal X}) \geq \tau^*_{r,u}$.
This means that no feasible solution in (\ref{eqn:weighted}) satisfies $\sum_{k=1}^m r_k x_k = n$, and hence
$u \not \in {\rm BLP}({\cal B})$.
Thus, we have the following result on the membership of the BL-polytope, which does not depend on the common denominator of an input vector $u$.
\begin{Thm}
For $\epsilon >0$, 
we can decide in $O(n^7m^3 \epsilon^{-2})$ time 
whether a given rational vector $u \in \QQ^m_{\geq 0}$ belongs to
${\rm BLP}_{\epsilon}({\cal B})$ or does not belong to ${\rm BLP}({\cal B})$.  
\end{Thm}

The BL-polytope ${\rm BLP}({\cal B})$ can be empty, that is, 
the BL-inequality (\ref{eqn:BL}) fails for every $u$ and every finite $C$.
This motivates us to consider the problem of deciding nonemptiness of the BL-polytope.
With $u ={\bf 1}$, apply Lemma~\ref{lem:BLP_epsilon} and Theorems~\ref{thm:main_cu} and \ref{thm:fpt_c}.
Then we have: 
\begin{Thm}
\begin{itemize}
\item[(1)] For $\epsilon >0$, 
we can decide in $O(n^7m^3 \epsilon^{-2})$ time 
whether ${\rm BLP}_{\epsilon}({\cal B}) \neq \emptyset$ or  
${\rm BLP}({\cal B}) = \emptyset$. 
\item[(2)] We can decide nonemptiness of 
${\rm BLP}({\cal B})$ in $O(n^{2n+7} m^3)$ time.
\end{itemize}
\end{Thm}

These appear be the first algorithms that run on a real RAM with real inputs $({\cal B},u)$.
However, for rational inputs, assumption (A) may not be appropriate from the viewpoint of bit-complexity,
and a polynomial-time algorithm in the standard RAM model would be desirable.
Our algorithms also run on the RAM model when $\KK = \QQ$ 
but they cannot be polynomial-time algorithms, unfortunately.
This is because
successive flag operations (Lemma~\ref{lem:n^3})  
can cause an exponential explosion of the bit-length of 
the basis representing flag ${\cal X}$, even if the number of iterations is polynomially bounded.
This bit-complexity issue arises in the nc-rank computation algorithm in~\cite{HamadaHirai} 
and its generalization in~\cite{HIOS2025}. 
In these settings, the rational arithmetic can be avoided 
by a reduction to finite fields of polynomial size. We do not know whether 
such a finite-field reduction is applicable to the BL-polytope problems.
Further analysis must be left to future research.

\section{Relation to set packing and covering LPs}\label{sec:set_packing/cover}

In this section, 
we verify that, when each $S_k$ is a coordinate subspace, 
(\ref{eqn:fractional_subspace_packing}) and (\ref{eqn:fractional_hitting_subpace}) 
reduce to the usual primal-dual pair of set packing and covering LPs.
For a subset $J \subseteq [n]$, let $\KK^{J}$ denote 
the coordinate subspace $\{ v \in \KK^n \mid v_i = 0\ (i \in [n] \setminus J)\}$.
\begin{Lem}\label{lem:coordinate_subspace}
Let $X \in {\cal S}(\KK^n)$, and 
let $A$ be an $n \times k$ matrix whose column vectors are a basis of $X$.  
If the submatrix of $A$ with row indices  $i_1,i_2,\ldots,i_k$ is nonsingular, 
then for any $J \subseteq [n]$ it holds 
\begin{equation}\label{eqn:J}
\dim X \cap \KK^{J} \leq |\{i_1,i_2,\ldots,i_k\} \cap J| = \dim \KK^{\{i_1,i_2,\ldots,i_k\}} \cap \KK^J.
\end{equation}
\end{Lem}

\begin{proof}
We can suppose that the first $\ell$ columns of $A$ form a basis of $X \cap \KK^{J}$, where 
$\dim X \cap \KK^{J} = \ell$.  
Then, $A_{ij} = 0$ holds for $(i,j) \in ([n] \setminus J) \times [\ell]$, and every maximal linearly independent subset of row vectors of $A$ must include 
$\ell$ rows in $J$. This implies~(\ref{eqn:J}).
%
%
%
\end{proof}
Therefore, if each $S_k$ is the coordinate subspace $\KK^{J_k}$ for $J_k \subseteq [n]$, 
then we can replace vector subspaces $X$ by the above coordinate subspaces $\KK^{\{i_1,i_2,\ldots,i_k\}}$
in 
(\ref{eqn:fractional_subspace_packing}), (\ref{eqn:fractional_hitting_subpace}), and the hitting subspace problem.
Since the inequality $\sum_{k} x_k |J_k \cap I| \leq |I|$ for $\KK^I$ is implied 
by $\sum_{k =1}^m x_k |J_k \cap \{i\}| \leq 1$ for $\KK^{\{i\}}$,  (\ref{eqn:fractional_subspace_packing}) is the set packing LP:
\begin{equation}\label{eqn:set-packing_LP}
\mbox{Max.}\quad {\bf 1}^{\top} x \quad \mbox{s.t.} \quad \sum_{k=1}^m {\bf 1}_{J_k}x_k \leq {\bf 1},\ x \in \RR^n_{\geq 0}. 
\end{equation}
One can also see that (\ref{eqn:fractional_hitting_subpace}) is equivalent to the set cover LP:
\begin{equation}\label{eqn:set-covering_LP}
\mbox{Min.}\quad {\bf 1}^{\top} p \quad \mbox{s.t.} \quad {\bf 1}_{J_k}^{\top} p \geq 1 \ (k \in [m]),\ p \in \RR^n_{\geq 0}. 
\end{equation}
To see this, from $\lambda$, 
consider $p$ by $p_i := \sum_{X \subseteq [n]: i \in X}\lambda(X)$.

This relation has the following implication for BL-polytopes.
\begin{Prop}
    Suppose that each full-row-rank matrix $B_k \in \RR^{r_k \times n}$
    is a coordinate projection; that is, there exists $J_k \subseteq [n]$ such that  $|J_k| = r_k$ and $(B_k)_{ij} = 0$ 
    for all $(i,j) \in [r_k] \times ([n] \setminus J_k )$.
    Then ${\rm BLP}({\cal B})$ is equal to 
    the perfect fractional set-packing polytope:
    \begin{equation}\label{eqn:perfect_set-packing}
    {\rm BLP}({\cal B}) = \left\{ u \in \RR^m_{\geq 0} \mid \sum_{k=1}^m {\bf 1}_{J_k} u_k = {\bf 1}\right\}.
     \end{equation}
\end{Prop}
Although the expression (\ref{eqn:perfect_set-packing}) appears to be new in the literature, 
the essential part~$(\supseteq)$ is (a part of) 
a result of Finner~\cite[Theorem 2.1]{Finner1992} on a generalization of H\"older's inequality. 
The reverse inclusion $(\subseteq)$ is immediate 
from the argument leading to (\ref{eqn:set-packing_LP}).
In particular, the membership problem for such BL-polytopes can be solved in polynomial time.

\section*{Acknowledgments}
The author thanks Yuya Goto, Yuni Iwamasa, Taihei Oki, and Tasuku Soma for comments,  
and Hiroshi Tsuji and Neal Bez for sharing Finner's result~\cite{Finner1992}.  
The author was supported 
by JSPS KAKENHI Grant Number JP24K21315 and JP26H01996.
The author used Microsoft Copilot to assist with literature survey, organization of ideas, and language editing.
The author reviewed and revised all outputs and takes full responsibility for the content of the manuscript.
\bibliographystyle{plain}
\bibliography{frac_packing}

\end{document}